\documentclass[11pt]{article}

\usepackage{amsmath,amssymb,mathtools,bm}
\usepackage{booktabs,array,tabularx}
\usepackage[font=small,labelfont=bf]{caption}
\usepackage{pgfplots}
\usepackage{flafter}
\usepackage{placeins}
\usepackage[margin=1in]{geometry}
\usepackage[numbers,sort&compress]{natbib}
\usepackage{amsthm}
\usepackage[hidelinks]{hyperref}
\usepackage{url}
\pgfplotsset{compat=1.18}
\theoremstyle{plain}
\newtheorem{theorem}{Theorem}[section]
\newtheorem{proposition}[theorem]{Proposition}
\newtheorem{corollary}[theorem]{Corollary}

\theoremstyle{definition}
\newtheorem{definition}[theorem]{Definition}

\newtheorem{example}[theorem]{Example}
\numberwithin{equation}{section}

\DeclareMathOperator{\Law}{Law}
\DeclareMathOperator{\MMD}{MMD}

\DeclareMathOperator{\argmin}{argmin}
\newcommand{\R}{\mathbb R}
\newcommand{\Z}{\mathbb Z}
\newcommand{\E}{\mathbb E}
\newcommand{\Prob}{\mathbb P}
\newcommand{\1}{\mathbf 1}
\newcommand{\calA}{\mathcal A}

\newcommand{\calC}{\mathcal C}
\newcommand{\calD}{\mathcal D}
\newcommand{\calE}{\mathcal E}

\newcommand{\calH}{\mathcal H}

\newcommand{\calQ}{\mathcal Q}
\newcommand{\calS}{\mathcal S}
\newcommand{\calT}{\mathcal T}
\newcommand{\calX}{\mathcal X}
\newcommand{\calY}{\mathcal Y}
\newcommand{\calW}{\mathcal W}
\newcommand{\full}{\mathrm{full}}
\newcommand{\quot}{\mathrm{quot}}

\begin{document}

\title{Causal inference for group-contaminated structured outcomes:\\
observable quotients, lossless reduction and exact randomization inference}
% Insert the complete author names and affiliations before public upload.
\author{
	Usef Faghihi$^{\dagger}$ \qquad Amir Saki$^{\dagger}$\\[0.5em]
	\small Département de mathématiques et d'informatique\\
	\small Université du Québec à Trois-Rivières\\
	\small Trois-Rivières, Québec, Canada\\[0.3em]
	\small $^{\dagger}$These authors contributed equally to this work.\\
	\small Usef Faghihi: \texttt{usef.faghihi@uqtr.ca}\\
	\small Amir Saki: \texttt{amir.saki@uqtr.ca}
}
\maketitle

\begin{abstract}
Structured potential outcomes such as microscopy images may be recorded after an unknown, unit-specific transformation. If that transformation can depend on treatment, covariates or the intrinsic outcome, raw-coordinate analyses may mix biological effects with acquisition geometry. We study the unrestricted observation model $X=\Gamma\cdot Y(A)$ and characterize its observable information: a target is uniformly recoverable exactly when it is constant on group orbits, while a Borel maximal invariant retains every measurable invariant target. We then distinguish observability from statistical losslessness. A quotient-faithful reconstruction theorem shows that quotient reduction is sufficient for the full transformed experiment exactly when the conditional law of the raw observation given treatment, covariates and the quotient has a parameter-free version. Conditional Haar contamination on a compact group yields Blackwell equivalence as a special case; it is not imposed in the main model. We also separate independent site-specific product actions from shared diagonal actions and show why componentwise canonicalization can discard relative cross-site information. Under explicit metric and kernel regularity, an approximate-contamination theorem bounds quotient-law Wasserstein error and the induced perturbation of population maximum mean discrepancy. For finite-support multichannel lattice images, we construct a maximal invariant under integer translations and quarter turns, combine its characteristic Gaussian kernel with a complete paired-swap test, and retain the original simulations and RxRx1 HUVEC study. Under the sharp null, the quotient test rejected in 0.052 of simulation replicates; at unit effect strength its power was 0.992. The primary RxRx1 contrast had an enumerated paired-swap $p$-value of 0.0078.
\end{abstract}


\section{Introduction}\label{sec:introduction}

Many experimental outcomes are structured objects rather than scalars. An image, shape, network or function may be represented in a coordinate system that varies from unit to unit. In microscopy, for example, translation or orientation can arise from acquisition and preprocessing rather than from the biological intervention. We study the observation model
\begin{equation}\label{eq:observation}
X=\Gamma\cdot Y(A),
\end{equation}
where $A$ is treatment, $Y(a)$ is the intrinsic structured potential outcome and $\Gamma$ is an unobserved group element. Crucially, no independence between $\Gamma$ and treatment, covariates or potential outcomes is imposed. Averaging over a posited nuisance distribution is therefore neither required nor generally available.

The first question is an information boundary: which targets of $Y(a)$ are determined by the transformed observation? Orbit invariance gives an exact answer. The second question is causal: under what design or exchangeability conditions does the observed quotient distribution identify its interventional counterpart? A third, logically stronger question is statistical: when does discarding the within-orbit coordinate lose no information about the full family of observed-data laws? Maximal invariance answers the first question, but not the third. Conflating those statements would overstate what quotienting achieves.

This paper develops a theorem-to-algorithm-to-application account of these three layers. The principal additions to the original outcome-level framework are:
\begin{enumerate}
\item a sharp sufficiency boundary based on a common conditional law and a quotient-faithful reconstruction kernel;
\item a conditional-Haar Blackwell-equivalence corollary, stated only as a special case of the unrestricted contamination framework;
\item a product-versus-diagonal action result that identifies when sitewise canonicalization is appropriate and when it erases relative configuration;
\item an approximate-contamination stability result under explicit orbit-metric and reproducing-kernel assumptions.
\end{enumerate}

The individual ingredients---maximal invariants, measurable factorization, statistical sufficiency, Blackwell comparison, Haar measure, characteristic kernels and Fisher randomization---are classical \citep{Blackwell1953,Eaton1989,Wijsman1990,Fisher1935}. The contribution claimed here is their outcome-level causal synthesis under unrestricted latent group contamination, together with an explicit noncompact lattice construction and a finite-sample inferential workflow. We do not claim novelty for those classical ingredients separately.

Existing causal methods for metric- and random-object outcomes generally take the structured outcome as directly observed \citep{Kurisu2024,Shin2024,Bhattacharjee2025}; registration-based functional methods address a related preprocessing problem \citep{Raykov2025}. Topological transforms can be injective on specified shape classes when sufficiently rich directional information is available \citep{Turner2014,Curry2022}. Our setting is complementary: the observation equation itself contains a latent, unit-specific action, and the target is the full observable orbit rather than a chosen low-dimensional summary.

The application uses two-site, six-channel RxRx1 wells. Under separate site-specific acquisition motions the nuisance group is a product of lattice rigid-motion groups. Translation normalization and finite rotational minimization yield a measurable maximal invariant without placing a probability law on the noncompact translation group. A Gaussian kernel on canonical representatives is characteristic for quotient laws, and complete enumeration of the reported paired assignment mechanism yields finite-sample inference for the quotient Fisher sharp null.

Section~\ref{sec:observable} develops observability and quotient-law identification. Section~\ref{sec:lossless} gives the statistical-losslessness boundary and Haar special case. Section~\ref{sec:multi} distinguishes product and diagonal actions. Section~\ref{sec:approx} proves stability under approximate contamination. Sections~\ref{sec:lattice} and~\ref{sec:randomization} give the image construction and exact test. Sections~\ref{sec:simulation} and~\ref{sec:rxrx1} report the preserved empirical results. Proof details that would interrupt the main argument appear in Appendix~\ref{sec:appendix}.

\section{Observable quotient targets and causal identification}\label{sec:observable}

\subsection{Measurable group contamination}

Let $\calY,\calC,\calS$ and $G$ be standard Borel spaces. Assume that $G$ is a measurable group and acts on $\calY$ through a jointly Borel map $(g,y)\mapsto g\cdot y$. The orbit of $y$ is $[y]_G=\{g\cdot y:g\in G\}$. A Borel map $\tau:\calY\to\calS$ is invariant if $\tau(g\cdot y)=\tau(y)$ for all $(g,y)$.

Let $C$ denote the observed pretreatment covariates, taking values in the standard Borel space $\calC$, and let $A$ denote the realized treatment assignment, taking values in a finite set $\calA$. For each $a\in\calA$, let $Y(a)$ be the intrinsic potential outcome under treatment $a$ in the potential-outcomes framework \citep{Rubin1974}. We observe $O=(C,A,X)$, where $X$ satisfies~\eqref{eq:observation}. The conditional distribution of $\Gamma$ given $C$, $A$ and the potential outcomes is unrestricted.

\begin{theorem}[Sharp observability]\label{thm:observability}
For a Borel target $\tau:\calY\to\calS$, the following are equivalent:
\begin{enumerate}
\item there is a Borel decoder $D:\calY\to\calS$ such that $D(g\cdot y)=\tau(y)$ for every $g\in G$ and $y\in\calY$;
\item $\tau$ is $G$-invariant.
\end{enumerate}
\end{theorem}

\begin{proof}
If $\tau$ is invariant, take $D=\tau$. Conversely, the identity element gives $D(z)=\tau(z)$ for every $z$. Hence $\tau(g\cdot y)=D(g\cdot y)=\tau(y)$.
\end{proof}

\begin{corollary}[Nonidentifiability outside the quotient]\label{cor:nonid}
If $\tau$ is not invariant, two data-generating processes satisfying~\eqref{eq:observation} can induce the same observed law while giving different laws of $\tau\{Y(a)\}$ for any fixed $a$.
\end{corollary}

\begin{proof}
Choose $y$ and $g$ with $\tau(g\cdot y)\ne\tau(y)$. In one process set $Y(a)=y$ and $\Gamma=g$; in the other set $Y(a)=g\cdot y$ and $\Gamma=e$. Both give $X=g\cdot y$, but their target values differ.
\end{proof}

\subsection{Maximal invariant and factorization}

\begin{definition}
A Borel map $M:\calY\to\calS$ is a maximal invariant when
\[
M(y)=M(z)\quad\Longleftrightarrow\quad z=g\cdot y\text{ for some }g\in G.
\]
\end{definition}

\begin{theorem}[Borel factorization]\label{thm:factorization}
Let $\calY,\calS,\calT$ be standard Borel spaces and $M:\calY\to\calS$ a Borel maximal invariant. Every Borel invariant $h:\calY\to\calT$ admits a Borel factorization $h=\bar h\circ M$. The restriction of $\bar h$ to $M(\calY)$ is unique.
\end{theorem}

The proof, given in Appendix~\ref{sec:appendix}, uses saturation of $h^{-1}(B)$, Lusin separation of two disjoint analytic images and the measurable factorization lemma; these are standard tools of descriptive set theory \citep{Kechris1995}. Thus $M$ retains every measurable invariant target; it does not follow that $M(X)$ is sufficient for a statistical parameter.

\subsection{Identification of the interventional quotient law}
For each $a\in\calA$, define the quotient potential outcome
\[
Q(a)=M\{Y(a)\},
\]
and define the observed quotient outcome by $Q=M(X)$. Invariance then gives
\begin{equation}\label{eq:quot-consistency}
	Q
	=
	M(X)
	=
	M\{\Gamma\cdot Y(A)\}
	=
	M\{Y(A)\}
	=
	Q(A)
\end{equation}
without imposing any restriction on the conditional distribution of $\Gamma$.

\begin{theorem}[Identification of the full quotient law]\label{thm:gformula}
Suppose, for every $a\in\calA$, that quotient consistency~\eqref{eq:quot-consistency} holds, $Q(a)\perp\!\!\!\perp A\mid C$, and $\Prob(A=a\mid C)>0$ almost surely. Then, for every Borel $B\subseteq\calS$,
\begin{equation}\label{eq:gformula}
\Prob\{Q(a)\in B\}=\int \Prob(Q\in B\mid A=a,C=c)\,dP_C(c).
\end{equation}
Consequently, the entire interventional law of the maximal invariant is identified.
\end{theorem}

\begin{proof}
Standard Borel spaces admit regular conditional probabilities \citep{Kallenberg2021}. Apply iterated expectation, conditional exchangeability, positivity and consistency in that order. Equality for every Borel $B$ determines the probability measure on $\calS$.
\end{proof}

By Theorem~\ref{thm:factorization}, the interventional law of any Borel invariant target is the pushforward of~\eqref{eq:gformula} through $\bar h$. This is the ordinary adjustment formula \citep{Robins1986,RosenbaumRubin1983} applied to a quotient-valued outcome, and it remains distinct from the sufficiency question considered next.

\section{When quotient reduction is statistically lossless}\label{sec:lossless}

Let $\Theta$ index a family of observed-data laws. We use the standard comparison-of-experiments framework \citep{LeCam1986,Torgersen1991}. Put
\[
O_{\full}=(C,A,X),\qquad O_{\quot}=T(O_{\full})=(C,A,M(X)),
\]
and write $\calE_{\full}=\{P_\theta^{\full}:\theta\in\Theta\}$ and $\calE_{\quot}=\{P_\theta^{\quot}:\theta\in\Theta\}$. Since $T$ is deterministic, the full experiment Blackwell-dominates the quotient experiment. The reverse comparison requires an additional condition.

\begin{definition}[Quotient-faithful kernel]
	A Markov kernel \(K\) from the quotient sample space to the full sample space
	is \emph{quotient-faithful} for \(\calE_{\quot}\) if
	\[
	K\!\left(s,T^{-1}(\{s\})\right)=1
	\qquad
	P_\theta^{\quot}\text{-almost surely for every }\theta\in\Theta.
	\]
	Thus, up to quotient-null sets under each member of the experiment, the kernel
	reconstructs a raw representative without changing the supplied covariates,
	treatment, or quotient value.
\end{definition}

\begin{theorem}[Boundary of lossless quotient reduction]\label{thm:boundary}
Assume the full and quotient sample spaces are standard Borel. The following are equivalent:
\begin{enumerate}
\item there is a Markov kernel $K$, independent of $\theta$, such that,
for every $\theta\in\Theta$, $K$ is a regular conditional
distribution of $O_{\full}$ given $O_{\quot}$ under
$P_\theta^{\full}$; equivalently,
\[
K(s,\cdot)
=
\Law_\theta(O_{\full}\mid O_{\quot}=s)
\quad
P_\theta^{\quot}\text{-almost surely}.
\]
\item there is a parameter-free quotient-faithful kernel $K$ such that
\begin{equation}\label{eq:reverse-kernel}
P_\theta^{\quot}K=P_\theta^{\full}\qquad\text{for every }\theta\in\Theta.
\end{equation}
\end{enumerate}
When these conditions hold, $O_{\quot}$ is sufficient for $\calE_{\full}$ and $\calE_{\full}\equiv_B\calE_{\quot}$. Hence every bounded decision problem has the same attainable risk set under the two experiments. If the family is dominated, the conditions are also equivalent to Fisher--Neyman factorization through $O_{\quot}$.
\end{theorem}

\begin{proof}
	Write \(Q_\theta:=P_\theta^{\quot}\). Suppose first that (1) holds. By the
	disintegration identity for the common regular conditional distribution, for
	every measurable set \(B\) in the full sample space,
	\[
	P_\theta^{\full}(B)
	=
	\int K(s,B)\,Q_\theta(ds)
	=
	(Q_\theta K)(B).
	\]
	Hence \(Q_\theta K=P_\theta^{\full}\) for every \(\theta\). Moreover, since
	\(O_{\quot}=T(O_{\full})\) deterministically, a regular conditional law of
	\(O_{\full}\) given \(O_{\quot}=s\) is supported on
	\(T^{-1}(\{s\})\), \(Q_\theta\)-almost surely. Thus \(K\) is
	quotient-faithful, proving (2).
	
	Conversely, suppose that (2) holds. Let \(B\) be measurable in the full sample
	space and let \(D\) be measurable in the quotient sample space. Quotient
	faithfulness implies that, \(Q_\theta\)-almost surely,
	\[
	K\!\left(s,B\cap T^{-1}(D)\right)
	=
	\mathbf 1_D(s)K(s,B).
	\]
	Indeed, conditional on input \(s\), the kernel is supported on
	\(T^{-1}(\{s\})\); this fibre is contained in \(T^{-1}(D)\) when \(s\in D\)
	and is disjoint from \(T^{-1}(D)\) when \(s\notin D\). Using
	\(Q_\theta K=P_\theta^{\full}\) with the measurable set
	\(B\cap T^{-1}(D)\) therefore gives
	\[
	\begin{aligned}
		P_\theta^{\full}\!\left(B\cap T^{-1}(D)\right)
		&=
		\int K\!\left(s,B\cap T^{-1}(D)\right)\,Q_\theta(ds)\\
		&=
		\int_D K(s,B)\,Q_\theta(ds).
	\end{aligned}
	\]
	Since \(T(O_{\full})=O_{\quot}\), this is precisely
	\[
	P_\theta^{\full}\!\left(O_{\full}\in B,\ O_{\quot}\in D\right)
	=
	\int_D K(s,B)\,P_\theta^{\quot}(ds),
	\]
	the defining identity for \(K\) to be a regular conditional distribution of
	\(O_{\full}\) given \(O_{\quot}\). The kernel is common to all
	\(\theta\), so (1) follows.
	
	The deterministic map \(T\) gives
	\(\calE_{\full}\succeq_B\calE_{\quot}\), while \(K\) gives
	\(\calE_{\quot}\succeq_B\calE_{\full}\). Hence the experiments are Blackwell
	equivalent. The decision-theoretic assertion follows from Blackwell
	equivalence, and the dominated factorization equivalence is the standard
	Fisher--Neyman sufficiency theorem.
\end{proof}

The quotient-faithfulness clause is essential for the stated equivalence. An arbitrary reverse kernel may reproduce each full marginal law while scrambling quotient values; such a kernel need not be a conditional distribution given the supplied quotient.

\begin{example}[Maximal invariance is not sufficiency]\label{ex:sign}
Let $G=\{-1,1\}$ act on $\{-1,1\}$ by multiplication. The maximal invariant is constant. Let $Y=1$ and let $\Prob_\theta(\Gamma=1)=\theta$, so $X=\Gamma$ and $\Prob_\theta(X=1)=\theta$. The raw observation identifies $\theta$, whereas the quotient law is identical for every $\theta$. Thus the conditional law of $X$ given the quotient depends on $\theta$ and the quotient is not sufficient. Under unrestricted contamination, within-orbit coordinates may carry information about the nuisance mechanism even though they carry no uniformly observable information about the intrinsic representative.
\end{example}

\subsection{Conditional-Haar contamination as a special case}

Let $G$ be compact and second countable, and let $\mu_G$ denote its
normalized Haar probability measure. Assume that $G$ acts jointly
Borel measurably on $\calY$ and that
$M:\calY\to\calS$ is a Borel maximal invariant. Put
\[
\calS_0=M(\calY),
\]
and assume that $\calS_0$ is a Borel subset of $\calS$. Let
$s:\calS_0\to\calY$ be a Borel section satisfying
\[
M\{s(m)\}=m
\qquad\text{for every }m\in\calS_0.
\]
For $m\in\calS_0$ and Borel $B\subseteq\calY$, define the orbit kernel
\begin{equation}\label{eq:haar-kernel}
	K_m(B)
	=
	\int_G
	\mathbf 1_B\{g\cdot s(m)\}\,d\mu_G(g).
\end{equation}
The kernel may be extended arbitrarily and measurably from $\calS_0$
to $\calS$; its values outside $\calS_0$ are irrelevant because
$M(X)\in\calS_0$ almost surely. Right invariance of Haar measure implies
that $K_m$ does not depend on the particular representative selected
from the orbit indexed by $m$.

\begin{corollary}[Conditional-Haar Blackwell equivalence]\label{cor:haar}
	Suppose that, for every \(\theta\in\Theta\),
	\begin{equation}\label{eq:cond-haar}
		\Law_\theta\{\Gamma\mid C,A,Y(A)\}=\mu_G
	\end{equation}
	almost surely. Then, for \(P_\theta^{\quot}\)-almost every
	\((C,A,m)\),
	\[
	\Law_\theta\{X\mid C,A,M(X)=m\}=K_m
	\]
	has a version that is independent of \(\theta\). Consequently,
	\[
	\calE_{\full}\equiv_B\calE_{\quot}.
	\]
\end{corollary}

\begin{proof}
	Fix a Borel set \(B\subseteq\calY\). The joint measurability of the action and
	the measurability of the section ensure that \(m\mapsto K_m(B)\) is Borel.
	Condition first on \((C,A,Y(A))\). Under~\eqref{eq:cond-haar},
	\[
	\begin{aligned}
		\Prob_\theta\{X\in B\mid C,A,Y(A)\}
		&=
		\int_G
		\mathbf 1_B\!\left(g\cdot Y(A)\right)\,d\mu_G(g).
	\end{aligned}
	\]
	Let \(y\in\calY\) and put \(m=M(y)\). Since \(s(m)\) and \(y\) lie in the same
	orbit, maximality of \(M\) yields an \(h\in G\) such that
	\[
	y=h\cdot s(m).
	\]
	Consequently,
	\[
	\int_G\mathbf 1_B(g\cdot y)\,d\mu_G(g)
	=
	\int_G\mathbf 1_B\!\left((gh)\cdot s(m)\right)\,d\mu_G(g).
	\]
	By right invariance of normalized Haar probability, \(gh\) has law \(\mu_G\)
	whenever \(g\) has law \(\mu_G\). Hence
	\[
	\int_G\mathbf 1_B(g\cdot y)\,d\mu_G(g)
	=
	\int_G\mathbf 1_B\!\left(g\cdot s(m)\right)\,d\mu_G(g)
	=
	K_m(B).
	\]
	It follows that
	\[
	\Prob_\theta\{X\in B\mid C,A,Y(A)\}
	=
	K_{M\{Y(A)\}}(B)
	\qquad\text{almost surely}.
	\]
	Taking conditional expectations given
	\(\sigma\{C,A,M(Y(A))\}\) therefore gives
	\[
	\begin{aligned}
		\Prob_\theta\{X\in B\mid C,A,M(Y(A))\}
		&=
		\mathbb E_\theta\!\left[
		K_{M\{Y(A)\}}(B)
		\mid C,A,M(Y(A))
		\right]\\
		&=
		K_{M\{Y(A)\}}(B)
		\qquad\text{almost surely}.
	\end{aligned}
	\]
	Since \(M\) is invariant and \(X=\Gamma\cdot Y(A)\),
	\[
	M(X)=M\{\Gamma\cdot Y(A)\}=M\{Y(A)\}
	\qquad\text{almost surely}.
	\]
	Thus
	\[
	\Law_\theta\{X\mid C,A,M(X)=m\}=K_m
	\]
	for \(P_\theta^{\quot}\)-almost every \((C,A,m)\), and this version is
	independent of \(\theta\).
	
	For completeness, define a kernel from the quotient sample space to the full
	sample space by
	\[
	\widetilde K\bigl((c,a,m),B\bigr)
	:=
	\int_{\calY}
	\mathbf 1_B(c,a,x)\,K_m(dx),
	\]
	for every measurable \(B\) in the full sample space. This kernel is
	parameter-free. Moreover, \(K_m\) is supported on
	\(M^{-1}(\{m\})\), since
	\[
	M\{g\cdot s(m)\}=M\{s(m)\}=m
	\qquad\text{for every }g\in G.
	\]
	Hence \(\widetilde K\) is quotient-faithful. It is a common regular
	conditional distribution of \(O_{\full}\) given \(O_{\quot}\), so
	Theorem~\ref{thm:boundary} yields
	\[
	\calE_{\full}\equiv_B\calE_{\quot}.
	\]
\end{proof}

The logical hierarchy is therefore:
\[
\begin{array}{ll}
\text{unrestricted }\Gamma: & M\{Y(a)\}\text{ is the maximal uniformly observable target};\\
\text{parameter-free within-orbit law:} & M(X)\text{ is statistically sufficient};\\
\text{conditional Haar }\Gamma: & \text{a concrete compact-group condition implying sufficiency}.
\end{array}
\]
The main results and empirical construction use the first line. The Haar assumption is not needed for observability, quotient identification or exact randomization inference.

\section{Independent product actions and shared diagonal actions}\label{sec:multi}

Structured units often contain labeled components, such as two microscopy sites. The appropriate group depends on whether acquisition transformations are component-specific or shared.

\begin{proposition}[Product versus diagonal actions]\label{prop:product-diagonal}
Let $M:\calY\to\calS$ be a maximal invariant for a $G$-action and let $J\ge2$. Define
\[
M^{\times J}(y_1,\ldots,y_J)=\bigl(M(y_1),\ldots,M(y_J)\bigr).
\]
\begin{enumerate}
\item Under the product action of $G^J$,
\[
(g_1,\ldots,g_J)\cdot(y_1,\ldots,y_J)=(g_1\cdot y_1,\ldots,g_J\cdot y_J),
\]
$M^{\times J}$ is a maximal invariant.
\item Under the diagonal action of $G$,
\[
g\cdot(y_1,\ldots,y_J)=(g\cdot y_1,\ldots,g\cdot y_J),
\]
$M^{\times J}$ is invariant. Moreover, on any diagonally
$G$-invariant model subset $\calD\subseteq\calY^J$, it is maximal
for the restricted diagonal action if and only if, whenever $y,z\in\calD$ and $z_j=g_j\cdot y_j$ for component-specific $g_j$, there is a single $g\in G$ satisfying $z_j=g\cdot y_j$ for every $j$.
\end{enumerate}
Thus componentwise canonicalization can discard relative cross-component information under a shared transformation.
\end{proposition}

\begin{proof}
For (1), suppose first that
\[
M^{\times J}(y_1,\ldots,y_J)
=
M^{\times J}(z_1,\ldots,z_J).
\]
Then $M(y_j)=M(z_j)$ for every $j$. By maximality of $M$, for each
$j$ there exists $g_j\in G$ such that $z_j=g_j\cdot y_j$. Therefore
\[
(z_1,\ldots,z_J)
=
(g_1,\ldots,g_J)\cdot(y_1,\ldots,y_J),
\]
so the two points belong to the same $G^J$-orbit. The converse follows
immediately from the invariance of $M$ in each component. Hence
$M^{\times J}$ is maximal for the product action.

For (2), invariance under the diagonal action follows because
\[
M(g\cdot y_j)=M(y_j)
\qquad\text{for every }j.
\]
Now let $y,z\in\calD$. Equality
$M^{\times J}(y)=M^{\times J}(z)$ is equivalent, by maximality of
$M$, to the existence of possibly different elements
$g_1,\ldots,g_J\in G$ satisfying
\[
z_j=g_j\cdot y_j,\qquad j=1,\ldots,J.
\]
The points $y$ and $z$ belong to the same diagonal orbit precisely
when these component-specific transformations can be replaced by a
single $g\in G$ satisfying $z_j=g\cdot y_j$ for every $j$. This is
exactly the stated condition.
\end{proof}

\begin{example}[Relative displacement]
Let $G=(\R,+)$ act on $\R$ by translation. A one-component maximal invariant is constant because the action is transitive. For a pair $(y_1,y_2)$ under the diagonal action, the difference $y_2-y_1$ is invariant and separates diagonal orbits. Componentwise quotienting is constant and therefore destroys this relative displacement. The image analogue is relative position or orientation between sites subjected to one common acquisition motion. In contrast, when the two sites undergo genuinely separate motions, relative pose is not uniformly observable and the product quotient is appropriate.
\end{example}

For the RxRx1 stress test below, the imposed transformation generator assigns separate site-specific motions, so the product action is the declared model. Proposition~\ref{prop:product-diagonal}(1) prevents that application-specific choice from being misread as a general prescription.

\section{Stability under approximate contamination}\label{sec:approx}

Exact invariance is action-specific. To obtain a correct perturbation result, additional metric regularity is needed. Let $(\calY,d)$ be a metric space on which $G$ acts by isometries and define the orbit pseudometric
\begin{equation}\label{eq:orbitmetric}
d_G(y,z)=\inf_{g\in G}d(y,g\cdot z).
\end{equation}
Let $(\calS,d_{\calS})$ contain a quotient representation $M$ satisfying
\begin{equation}\label{eq:lipschitzM}
d_{\calS}\{M(y),M(z)\}\le L_M d_G(y,z).
\end{equation}
Let $\mathcal P_1(\calS)$ denote the Borel probability measures on
$\calS$ having a finite first moment with respect to $d_{\calS}$.
For a measurable positive-semidefinite kernel $k$ with reproducing
kernel Hilbert space $\calH_k$ and Borel feature map
$\Phi:\calS\to\calH_k$, define
\[
\mu_k(P)=\int_{\calS}\Phi(s)\,P(ds)
\]
whenever this Bochner integral exists, and put
\[
\MMD_k(P,Q)
=
\|\mu_k(P)-\mu_k(Q)\|_{\calH_k}.
\]
\begin{theorem}[Quotient-law and MMD stability]\label{thm:stability}
For arms $a\in\{0,1\}$, let $(X_a,Y_a)$ be a coupling of the approximately contaminated and intrinsic outcomes, and assume
\[
\E d_G(X_a,Y_a)\le\varepsilon_a.
\]
Write
\[
\widetilde P_a=\Law\{M(X_a)\},
\qquad
P_a=\Law\{M(Y_a)\},
\]
and assume that $\widetilde P_a,P_a\in\mathcal P_1(\calS)$. Then
\begin{equation}\label{eq:w1bound}
W_{1,d_{\calS}}(\widetilde P_a,P_a)\le L_M\varepsilon_a.
\end{equation}
If the mean embeddings of $\widetilde P_a$ and $P_a$ exist and the
feature map satisfies
\begin{equation}\label{eq:featurelip}
\|\Phi(s)-\Phi(t)\|_{\calH_k}\le L_k d_{\calS}(s,t),
\end{equation}
then
\begin{equation}\label{eq:mmdbound}
\left|\MMD_k(\widetilde P_1,\widetilde P_0)-\MMD_k(P_1,P_0)\right|
\le L_kL_M(\varepsilon_1+\varepsilon_0).
\end{equation}
If the perturbation bounds hold almost surely, the same conclusions follow with those uniform bounds.
\end{theorem}

\begin{proof}
The supplied coupling and~\eqref{eq:lipschitzM} give
\[
W_{1,d_{\calS}}(\widetilde P_a,P_a)
\le\E d_{\calS}\{M(X_a),M(Y_a)\}
\le L_M\varepsilon_a.
\]
Let $\mu(R)=\E_{S\sim R}\Phi(S)$. By the same coupling and~\eqref{eq:featurelip},
\[
\|\mu(\widetilde P_a)-\mu(P_a)\|_{\calH_k}\le L_kL_M\varepsilon_a.
\]
Since $\MMD_k(R,S)=\|\mu(R)-\mu(S)\|_{\calH_k}$, the reverse triangle inequality followed by the last two bounds proves~\eqref{eq:mmdbound}.
\end{proof}

For a Gaussian kernel $k(s,t)=\exp\{-\|s-t\|^2/(2\sigma^2)\}$ on a Hilbert quotient embedding,
\[
\|\Phi(s)-\Phi(t)\|^2=2\{1-k(s,t)\}\le\|s-t\|^2/\sigma^2,
\]
so $L_k=1/\sigma$. The theorem is intentionally conditional: the lexicographic lattice canonicalizer in Section~\ref{sec:lattice} need not be globally Lipschitz near support changes or canonicalization ties. The approximate-action experiment in Table~\ref{tab:approx} therefore remains an empirical sensitivity analysis; it is not retroactively certified by Theorem~\ref{thm:stability} unless~\eqref{eq:lipschitzM} is verified on the relevant image class.

\section{A maximal invariant for lattice microscopy images}\label{sec:lattice}

\subsection{Lattice rigid motions}

Fix $L<\infty$. Let $\calX_{q,L}$ contain $q$-channel, 8-bit functions on $\Z^2$ with finite support whose nonempty bounding box has at most $L$ rows and columns, together with the zero image. Let $C_4=\{0,1,2,3\}$, with addition understood modulo four, and let
$R_r$ denote the counterclockwise rotation of $\Z^2$ through
$r\pi/2$. The action of $C_4$ on $\Z^2$ is $u\mapsto R_ru$. Thus the
semidirect product
\[
G_{\rm rig}=\Z^2\rtimes C_4
\]
has multiplication and inverse
\[
(t,r)(u,s)=\bigl(t+R_ru,r+s\bigr),
\qquad
(t,r)^{-1}=\bigl(-R_r^{-1}t,-r\bigr).
\]
Its action on an image $x$ is
\begin{equation}\label{eq:rigaction}
	\{(t,r)\cdot x\}(v)
	=
	x\{R_r^{-1}(v-t)\},
	\qquad v\in\Z^2.
\end{equation}
Here $t,u,v\in\Z^2$ are lattice coordinates or translations, whereas
$r,s\in C_4$ index quarter turns.

For nonzero $x$, let $b(x)$ be the lower corner of its support bounding box and set $N(x)=(-b(x),0)\cdot x$. For each $r\in C_4$, form $x_r=N\{(0,r)\cdot x\}$. The key $\kappa(x_r)$ is the finite vector containing the
bounding-box height and width, followed by all channel bytes in a
fixed site--channel--row--column order. These finite vectors are
compared using the lexicographic order. Define
\begin{equation}\label{eq:canonical}
c(x)=x_{r^\star},\qquad r^\star=\min\argmin_{r\in C_4}\kappa(x_r),
\end{equation}
and set $c(0)=0$.

\begin{theorem}[Lattice rigid-motion canonicalization]\label{thm:lattice}
The map $c$ is Borel and is a maximal invariant for $\Z^2\rtimes C_4$.
\end{theorem}

\begin{proof}
	Because $\calX_{q,L}$ is a countable union of finite-dimensional
	finite-valued image spaces, the support map, bounding-box map,
	translation normalization, finite serialization and lexicographic
	comparison operations are Borel. Since the minimization in
	\eqref{eq:canonical} is over the finite set $C_4$, the selected index
	$r^\star$ and hence $c$ are Borel.
	
	We next prove invariance. Translation normalization removes every
	integer translation:
	\[
	N\{(t,0)\cdot z\}=N(z)
	\qquad
	\text{for every }t\in\Z^2.
	\]
	Let $y=(t,s)\cdot x$. For any $r\in C_4$, the semidirect-product law
	gives
	\[
	(0,r)(t,s)=(R_rt,r+s).
	\]
	Therefore
	\[
	\begin{aligned}
		N\{(0,r)\cdot y\}
		&=
		N\{(0,r)(t,s)\cdot x\}\\
		&=
		N\{(R_rt,r+s)\cdot x\}\\
		&=
		N\{(0,r+s)\cdot x\}.
	\end{aligned}
	\]
	Thus the four normalized rotational candidates of $y$ are exactly
	the four normalized rotational candidates of $x$, with their indices
	permuted by $r\mapsto r+s$. Their lexicographically smallest elements
	are consequently equal, and hence
	\[
	c(y)=c(x).
	\]
	
	It remains to prove maximality. Suppose that $c(x)=c(y)=z$. By the
	construction of $c$, there exist translations $u,v\in\Z^2$ and
	rotations $r,s\in C_4$ such that
	\[
	z=(u,r)\cdot x
	\qquad\text{and}\qquad
	z=(v,s)\cdot y.
	\]
	Applying $(v,s)^{-1}$ to the second equality and substituting the
	first gives
	\[
	y
	=
	(v,s)^{-1}\cdot z
	=
	(v,s)^{-1}(u,r)\cdot x.
	\]
	Because $(v,s)^{-1}(u,r)\in G_{\rm rig}$, the images $x$ and $y$
	belong to the same $G_{\rm rig}$-orbit. Together with invariance, this
	shows
	\[
	c(x)=c(y)
	\quad\Longleftrightarrow\quad
	y=g\cdot x\text{ for some }g\in G_{\rm rig},
	\]
	so $c$ is a maximal invariant.
\end{proof}

An RxRx1 well consists of two labeled sites. Let
\[
\calW=\calX_{q,L}^2
\]
be the raw well space, with a typical well written as
$w=(x_1,x_2)$. Under the empirically imposed separate site-specific
motions, the well-level nuisance group is
\[
G_{\rm well}=G_{\rm rig}\times G_{\rm rig}.
\]
Define the well-level canonicalization map
\begin{equation}\label{eq:wellcanonical}
	M_{\rm well}(x_1,x_2)
	=
	\bigl(c(x_1),c(x_2)\bigr),
\end{equation}
and put
\[
\calW_{\rm can}=M_{\rm well}(\calW).
\]
Theorem~\ref{thm:lattice} and
Proposition~\ref{prop:product-diagonal}(1) imply that
$M_{\rm well}$ is a maximal invariant for the product action. No
probability measure is introduced on the noncompact translation
group.

\subsection{Characteristic quotient kernel and Euler signature}
Let
\[
\psi:\calW_{\rm can}\longrightarrow\mathbb R^{2qL^2}
\]
place the two canonical crops on fixed $L\times L$ zero canvases and
concatenate their entries in a fixed site--channel--row--column order.
This map is a Borel injection.
 For $\sigma>0$, define
\begin{equation}\label{eq:kernel}
	k_\sigma(w,w')
	=
	\exp\left[
	-\frac{
		\|\psi\{M_{\rm well}(w)\}
		-\psi\{M_{\rm well}(w')\}\|_2^2
	}{2\sigma^2}
	\right].
\end{equation}

\begin{proposition}\label{prop:kernel}
	The kernel~\eqref{eq:kernel} is positive semidefinite and invariant
	in each argument under the action of $G_{\rm well}$. Its kernel mean
	embedding distinguishes all Borel probability laws on
	$\calW_{\rm can}$.
\end{proposition}

\begin{proof}
	Let
	\[
	k_{\rm G}(u,v)
	=
	\exp\left\{-\frac{\|u-v\|_2^2}{2\sigma^2}\right\},
	\qquad
	u,v\in\mathbb R^{2qL^2}.
	\]
	The kernel $k_{\rm G}$ is positive semidefinite and characteristic on
	Euclidean space \citep{Sriperumbudur2010}. Equation~\eqref{eq:kernel}
	is its pullback through the Borel map
	$\psi\circ M_{\rm well}$, and is therefore positive semidefinite.
	
	For every $g\in G_{\rm well}$,
	\[
	M_{\rm well}(g\cdot w)=M_{\rm well}(w).
	\]
	Consequently,
	\[
	k_\sigma(g\cdot w,w')=k_\sigma(w,w')
	\quad\text{and}\quad
	k_\sigma(w,g\cdot w')=k_\sigma(w,w'),
	\]
	which proves invariance in each argument.
	
	Finally, let $P$ and $Q$ be Borel probability measures on
	$\calW_{\rm can}$. Equality of their kernel mean embeddings implies,
	by the characteristic property of the Gaussian kernel,
	\[
	\psi_{\#}P=\psi_{\#}Q.
	\]
	Because $\psi$ is a Borel injection between standard Borel spaces,
	the Lusin--Souslin theorem implies that $\psi(\calW_{\rm can})$ is
	Borel and that $\psi^{-1}$ is Borel on this image. Applying
	$\psi^{-1}$ to the two equal pushforward laws therefore yields
	$P=Q$.
\end{proof}

The secondary endpoint is a finite directional Euler vector. After canonicalization, Otsu thresholding \citep{Otsu1979} of the nuclear channel and removal of components smaller than four pixels produce a planar cubical complex. Euler characteristic is evaluated along 25 shape-adaptive thresholds in eight lattice directions at both sites. Because the calculation is applied after $M_{\rm well}$, the vector is exactly invariant under $G_{\rm well}$, but it is neither claimed nor used as a maximal invariant.

\section{Exact design-based testing}\label{sec:randomization}

Suppose the confirmation experiment contains $B$ paired blocks. In
block $b$, two eligible wells receive conditions $0$ and $1$, once
each. Order the two wells deterministically before observing their
assigned conditions, and let
\[
q_{b,0},q_{b,1}\in\calW_{\rm can}
\]
denote the quotient outcomes of the first and second wells,
respectively.

Let $Z_b\in\{0,1\}$ be the label of the condition assigned to the
first well. Thus, if $Z_b=1$, the first well receives condition $1$
and the second receives condition $0$; if $Z_b=0$, these assignments
are reversed. Conditional on the unordered eligible wells, their
complete potential outcomes and the pretreatment design information,
assume that
\[
Z=(Z_1,\ldots,Z_B)
\]
is uniform on
\[
\Omega=\{0,1\}^B.
\]

For an assignment $z\in\Omega$, the quotient samples assigned to
conditions $1$ and $0$ are, respectively,
\[
\calQ_1(z)=\{q_{b,1-z_b}:b=1,\ldots,B\},
\qquad
\calQ_0(z)=\{q_{b,z_b}:b=1,\ldots,B\}.
\]
Let $P_1$ and $P_0$ denote the corresponding population
quotient-outcome laws under conditions $1$ and $0$. Their population
kernel discrepancy is
\begin{equation}\label{eq:mmd}
	\MMD_k^2(P_1,P_0)
	=
	\E k(U,U')
	+
	\E k(V,V')
	-
	2\E k(U,V),
\end{equation}
where $U,U'$ are independent with law $P_1$, $V,V'$ are independent
with law $P_0$, and the four variables are mutually independent. The empirical statistic is the conventional two-sample $U$-statistic
\begin{align}\label{eq:ummd}
\widehat{\MMD}_u^2(z)
={}&\frac{1}{B(B-1)}\sum_{b\ne b'}k(q_{b,1-z_b},q_{b',1-z_{b'}})
+\frac{1}{B(B-1)}\sum_{b\ne b'}k(q_{b,z_b},q_{b',z_{b'}})\nonumber\\
&-\frac{2}{B^2}\sum_{b,b'}k(q_{b,z_b},q_{b',1-z_{b'}}).
\end{align}
Matched-block dependence can remove the usual unbiasedness interpretation; here~\eqref{eq:ummd} is a prespecified discrepancy statistic. The kernel two-sample construction follows the standard MMD formulation \citep{Gretton2012}, but the exactness below does not require unbiased estimation of a superpopulation parameter. Let $\calQ_{\rm pool}=\{q_{b,j}:b=1,\ldots,B,\ j\in\{0,1\}\}$.
The Gaussian bandwidth is fixed before examining the condition labels
by setting
\[
\sigma^2
=
\frac12\,
\operatorname{median}
\left\{
\|\psi(q)-\psi(q')\|_2^2:
q,q'\in\calQ_{\rm pool},\
q\ne q',\
\|\psi(q)-\psi(q')\|_2>0
\right\}.
\]
Because $\calQ_{\rm pool}$ does not change under paired label swaps,
the bandwidth is constant over the randomization distribution.

Let $S:\Omega\to\mathbb R$ be a deterministic test statistic; for the
primary analysis, $S(z)=\widehat{\MMD}_u^2(z)$. Define the exact
upper-tail randomization $p$-value by
\begin{equation}\label{eq:pvalue}
	p(Z)
	=
	2^{-B}
	\sum_{z\in\Omega}
	\mathbf 1\{S(z)\ge S(Z)\}.
\end{equation}

\begin{theorem}[Finite-sample randomization validity]\label{thm:randomization}
Assume well-defined reagent versions, no cross-well interference, swap-invariant block retention and conditional uniformity of $Z$ on $\Omega$. Under the Fisher sharp null that every confirmation well's invariant potential outcome is unchanged by the two conditions, the test rejecting when $p(Z)\le\alpha$ has conditional rejection probability at most $\alpha$ for every $\alpha\in[0,1]$.
\end{theorem}

\begin{proof}
Under the sharp null, all quotient outcomes and label-invariant tuning parameters are fixed over $\Omega$. Let $N=2^B$ and $r(z)=\#\{u\in\Omega:S(u)\ge S(z)\}$. Then $p(z)=r(z)/N$. For each integer $m$, at most $m$ assignments have $r(z)\le m$; ties can only increase the upper-tail rank. Taking $m=\lfloor\alpha N\rfloor$ and using uniformity proves the result.
\end{proof}

The theorem is exact conditional on the reported assignment mechanism and is an instance of finite-sample randomization inference \citep{Fisher1935,LehmannRomano2005}. The metadata audit can establish that paired siRNAs occur within the same plate randomization set; it cannot itself prove that the original allocation algorithm was uniform. The primary contrast uses one quotient-pixel test. Nine additional contrasts form a separate secondary family adjusted by Holm's method \citep{Holm1979}.

\section{Simulation study}\label{sec:simulation}

Each of 250 replicates at each effect strength contains eight randomized blocks, two units per block, two sites per unit, six channels and $28\times28$ intrinsic lattice arrays. Treatment adds a concentric annulus with strength $\eta\in\{0,0.35,0.70,1.00\}$. At $\eta=0$, the two treatment potential outcomes of each simulated
unit are bytewise identical; different units may nevertheless have
different intrinsic morphologies. After assignment, each site receives an integer shift and quarter turn from a deterministic generator depending on treatment, realized morphology, unit identity and a frozen replicate seed. The analyzer receives neither the group elements nor their generator inputs. All $2^8=256$ paired assignments are enumerated.

\begin{table}[t]
\caption{Paired-swap rejection fractions under treatment- and outcome-dependent acquisition transformations. The interval at $\eta=0$ is the exact 95\% Clopper--Pearson interval \citep{ClopperPearson1934} over 250 independent replicates.\label{tab:simulation}}
\centering\small
\begin{tabular}{@{}lccccc@{}}
\toprule
Method & $\eta=0$ & $0.35$ & $0.70$ & $1.00$ & Assignments\\
\midrule
Exact quotient & 0.052 (0.028, 0.087) & 0.140 & 0.848 & 0.992 & 256\\
Raw pixels & 1.000 (0.985, 1.000) & 1.000 & 1.000 & 1.000 & 256\\
Moment registration & 0.056 (0.031, 0.092) & 0.540 & 0.996 & 1.000 & 256\\
Finite Euler signature & 0.040 (0.019, 0.072) & 0.184 & 1.000 & 1.000 & 256\\
\bottomrule
\end{tabular}
\end{table}

For the exact quotient and post-canonical Euler endpoints, positive-$\eta$ columns estimate power against a known intrinsic effect. Raw-pixel and moment-registration values are diagnostic rejection fractions because those endpoints are not invariant under the informative acquisition mechanism.

\begin{figure}[!htbp]
\centering
\begin{tikzpicture}
\begin{axis}[
width=.82\linewidth,height=5.2cm,
xlabel={Intrinsic morphological effect strength},
ylabel={Paired-swap rejection fraction},
xmin=0,xmax=1,ymin=0,ymax=1.03,
legend style={font=\scriptsize,at={(0.02,0.98)},anchor=north west},
grid=major]
\addplot+[mark=*,thick] coordinates {(0,0.052) (0.35,0.140) (0.70,0.848) (1,0.992)};
\addlegendentry{Exact quotient}
\addplot+[mark=square*,thick] coordinates {(0,1) (0.35,1) (0.70,1) (1,1)};
\addlegendentry{Raw pixels}
\addplot+[mark=triangle*,thick] coordinates {(0,0.056) (0.35,0.540) (0.70,0.996) (1,1)};
\addlegendentry{Moment registration}
\addplot+[mark=diamond*,thick] coordinates {(0,0.040) (0.35,0.184) (0.70,1) (1,1)};
\addlegendentry{Euler secondary}
\addplot+[dashed] coordinates {(0,0.05) (1,0.05)};
\addlegendentry{Nominal 0.05}
\end{axis}
\end{tikzpicture}
\caption{Simulation rejection fractions. Finite-sample causal validity
	under informative acquisition pertains to the quotient and
	post-canonical Euler endpoints; the other methods are diagnostics.
	The exact quotient begins near the nominal 0.05 level and rises to
	0.992, raw pixels remain at 1.0, and moment registration and the Euler
	endpoint begin near 0.05 and approach 1.0 as effect strength increases.
	\label{fig:simulation}}
\end{figure}
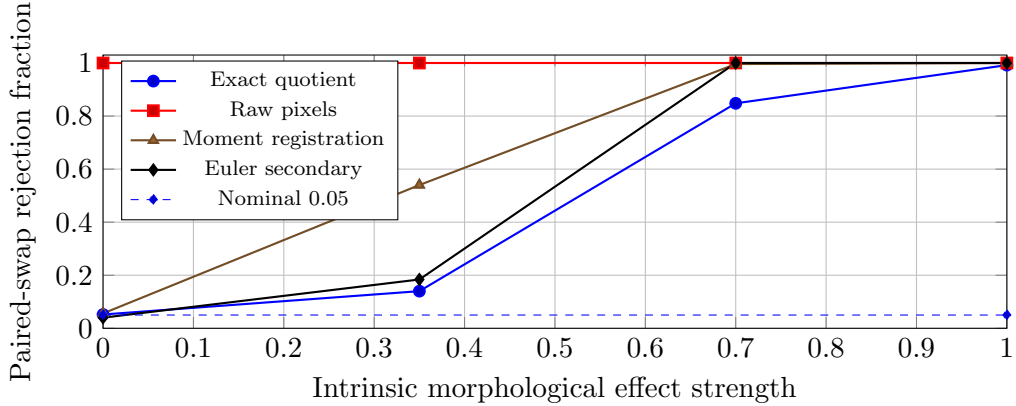

\FloatBarrier
\section{RxRx1 confirmation experiment}\label{sec:rxrx1}

\subsection{Units, selection and imposed acquisition}

The public RxRx1 release \citep{Recursion2023,Sypetkowski2023} contains 125,510 six-channel $512\times512$ images, with two nonoverlapping sites per well, 51 experimental batches, four cell types and 1,138 siRNAs, including 1,108 noncontrol perturbations. HUVEC contributes 24 batches. The well is the experimental unit; sites and pixels are components of one structured outcome, not independent replicates. The conditions are applied siRNA reagents, so the estimand is reagent-specific rather than automatically gene-specific.

HUVEC-01 through HUVEC-12 are used for discovery, HUVEC-13 through HUVEC-16 are reserved, and HUVEC-17 through HUVEC-24 supply eight confirmation blocks. Selection uses public 128-dimensional embeddings from discovery rows only. A split-half score combines within-half signal-to-noise with agreement of contrast direction. Eligibility requires one two-site well per condition in every discovery batch and identical plate signatures. Greedy matching yields ten disjoint pairs; the highest score defines the primary pair and nine pairs form the secondary family. The row split does not prove that the public embedding generator was trained independently of every confirmation image, so complete algorithmic independence remains an assumption.

The frozen confirmation subset has 160 wells, 320 sites and 1,920 single-channel PNG files. Each site is downsampled to $64\times64$, retained as 8-bit data and zero-padded. For each published nuisance seed, a BLAKE2b digest of the seed, well identifier and site determines a quarter turn and a signed translation. The transformation also depends on treatment and a coordinate-free morphology score. All channels at one site share a motion, while the two sites receive separate motions. Parameters are logged for audit but excluded from every representation and test routine; the software aborts if nonzero support is cropped.

\subsection{Confirmatory results}

The primary contrast is siRNA 384 (s38230) versus siRNA 747 (s37149). The completeness and same-plate gates passed for all frozen pairs and blocks. The discovery-only manifest has SHA-256 digest \texttt{b6f350ca96fb9e7a6271845dd6e432ed9dca556144a1a332d8f5af168642b968}.

\begin{table}[!htbp]
\caption{Primary RxRx1 pair under imposed lattice rigid motions. The quotient-pixel endpoint is the sole primary test; Euler is secondary; moment registration and raw pixels are acquisition-sensitive diagnostics.\label{tab:primary}}
\centering
\begin{tabular}{@{}lrr@{}}
\toprule
Representation & $\widehat{\MMD}_u^2$ & Paired-swap $p$\\
\midrule
Canonical quotient pixels & 0.0036 & 0.0078\\
Finite directional Euler signature & 0.7535 & 0.0078\\
Moment registration & -0.0065 & 0.2109\\
Raw transformed pixels & 0.1664 & 0.0078\\
\bottomrule
\end{tabular}
\end{table}

No secondary contrast remained significant after Holm adjustment. Exact invariance held across all acquisition seeds: maximum clean-versus-contaminated feature discrepancies were zero for the canonical pixels and Euler signature, as were the largest changes in quotient $\widehat{\MMD}_u^2$ and exact $p$-value. The algebraic and software audit found zero canonical-code failures among 6,400 generated transformations; 29 software checks passed.

\begin{table}[!htbp]
\caption{Canonical-quotient tests for the nine prespecified secondary contrasts. These contrasts are not biological replications of the primary contrast.\label{tab:secondary}}
\centering\scriptsize
\begin{tabular}{@{}clrrr@{}}
\toprule
Pair & Conditions & $\widehat{\MMD}_u^2$ & Swap $p$ & Holm $p$\\
\midrule
1 & s21474 vs s21486 & 0.0040 & 0.0859 & 0.1719\\
2 & s29342 vs s28343 & 0.0045 & 0.0156 & 0.0703\\
3 & s18592 vs s18763 & 0.1266 & 0.0078 & 0.0703\\
4 & s20782 vs s21523 & 0.0031 & 0.0078 & 0.0703\\
5 & s20401 vs s19178 & 0.0228 & 0.0156 & 0.0703\\
6 & s29371 vs s28116 & 0.0063 & 0.0078 & 0.0703\\
7 & s17831 vs s20468 & 0.0486 & 0.0078 & 0.0703\\
8 & s37483 vs s224918 & -0.0107 & 0.1562 & 0.1719\\
9 & s37712 vs s224940 & 0.0761 & 0.0078 & 0.0703\\
\bottomrule
\end{tabular}
\end{table}

\begin{table}[t]
\caption{Acquisition-only sharp-null stress test over frozen pairs and nuisance seeds. Intervals are descriptive Clopper--Pearson summaries; pair--seed tests share biological blocks.\label{tab:sharpnull}}
\centering\small
\begin{tabular}{@{}lrrl@{}}
\toprule
Method & Tests & Rejections & Fraction (95\% interval)\\
\midrule
Exact quotient & 200 & 0 & 0.000 (0.000, 0.018)\\
Finite Euler signature & 200 & 0 & 0.000 (0.000, 0.018)\\
Moment registration & 200 & 0 & 0.000 (0.000, 0.018)\\
Raw pixels & 200 & 200 & 1.000 (0.982, 1.000)\\
\bottomrule
\end{tabular}
\end{table}

\begin{table*}[t]
\caption{Approximate-action stress test for the primary pair. Feature error is relative to the exactly transformed baseline. These are empirical sensitivity results; Theorem~\ref{thm:stability} applies only when its Lipschitz assumptions are verified.\label{tab:approx}}
\centering\scriptsize
\begin{tabular}{@{}llrrr@{}}
\toprule
Perturbation & Method & Relative feature error & $|\Delta\widehat{\MMD}_u^2|$ & $|\Delta p|$\\
\midrule
Rotation $5^\circ$ & Exact quotient & 0.8432 & 0.0012 & 0.0000\\
 & Euler signature & 0.0779 & 0.0860 & 0.0000\\
 & Moment registration & 0.7295 & 0.0014 & 0.0234\\
 & Raw pixels & 0.6540 & 0.0013 & 0.0000\\
Rotation $10^\circ$ & Exact quotient & 0.8967 & 0.0022 & 0.0000\\
 & Euler signature & 0.0827 & 0.0530 & 0.0000\\
 & Moment registration & 0.8124 & 0.0021 & 0.1250\\
 & Raw pixels & 0.8038 & 0.0012 & 0.0000\\
Rotation $20^\circ$ & Exact quotient & 0.9715 & 0.0005 & 0.0000\\
 & Euler signature & 0.1076 & 0.0898 & 0.0000\\
 & Moment registration & 0.8662 & 0.0009 & 0.0469\\
 & Raw pixels & 0.8802 & 0.0014 & 0.0000\\
Two-pixel support crop & Exact quotient & 0.8253 & 0.0026 & 0.0000\\
 & Euler signature & 0.1478 & 0.0164 & 0.0000\\
 & Moment registration & 0.6857 & 0.0203 & 0.2031\\
 & Raw pixels & 0.3523 & 0.0001 & 0.0000\\
Gaussian noise & Exact quotient & 0.5949 & 0.0007 & 0.0000\\
 & Euler signature & 0.0618 & 0.0244 & 0.0000\\
 & Moment registration & 0.3870 & 0.0067 & 0.1328\\
 & Raw pixels & 0.1775 & 0.0006 & 0.0000\\
\bottomrule
\end{tabular}
\end{table*}

\FloatBarrier
\section{Discussion}\label{sec:discussion}

Under unrestricted contamination, invariance is an observability requirement rather than a modeling preference. A noninvariant target can change while the observed law remains fixed; a maximal invariant captures the entire recoverable orbit. Causal identification then requires a separate design or exchangeability argument. The losslessness theorem adds a third distinction: even a maximal invariant may discard information about how probability is distributed within an orbit. Only a parameter-free conditional law, or a condition such as conditional Haar randomization that implies one, makes quotienting sufficient for the full statistical experiment.

The distinction between product and diagonal actions has practical consequences. Sitewise canonicalization is correct for the imposed RxRx1 mechanism because each site receives a separate motion. If a microscope applies one common transformation to all sites, a diagonal quotient should retain relative placement and orientation. Treating that problem as a product action would erase scientifically meaningful relational structure.

The approximate-contamination theorem links orbit error to quotient-law and MMD error, but it also clarifies the limits of the current stress test. Lexicographic canonicalization can be discontinuous. Large feature changes accompanied by stable $p$-values in Table~\ref{tab:approx} demonstrate test-level robustness for the frozen perturbations, not global representation stability. Establishing a Lipschitz quotient embedding on a scientifically relevant image class is a separate research problem.

The RxRx1 analysis is a controlled unknown-acquisition experiment on real biological images. Its content and batch structure are real; the transformation law is imposed and exactly audited. The causal conclusion is conditional on the reported within-plate assignment mechanism, well-defined reagent versions, no cross-well interference and swap-invariant retention. The primary result concerns a reagent-condition contrast, not automatically a gene-level mechanism. With eight pairs, exact enumeration is a strength, but the randomization distribution has coarse resolution.

\section{Conclusion}\label{sec:conclusion}

Unknown transformations need not be estimated to conduct causal inference on the information that survives them. Under a measurable group action, orbit invariance is necessary and sufficient for uniform observability, and a maximal invariant contains every measurable observable target. Quotient reduction is statistically lossless only under the additional, exact boundary of a parameter-free conditional law given the quotient; conditional Haar contamination is one useful special case. Correct specification of product or diagonal actions determines whether relative component information is preserved. The lattice construction, characteristic quotient kernel and complete paired randomization test turn these distinctions into an auditable workflow for structured outcomes under informative acquisition.

\section*{Data availability}
The analysis uses the public RxRx1 dataset subject to its stated licence. The accompanying software is described by the reproducibility statement below. Repository or archival identifiers will be added when the public code archive is released.

\section*{Reproducibility statement}
The analysis software downloads the frozen RxRx1 subset, verifies file and manifest hashes, constructs all representations, enumerates the complete assignment distribution, performs multiplicity adjustment and sensitivity analyses, executes unit tests and writes the numerical outputs consumed by the manuscript. The source fingerprint reported for the analysed version is \texttt{source-\allowbreak e9f042c387cf}.

\section*{Author contributions}
Usef Faghihi and Amir Saki contributed equally to the conception,
methodology, theoretical development, analysis, interpretation of
results and preparation of the manuscript. Both authors reviewed and
approved the final manuscript.

\section*{Funding}
This research received no specific grant from any funding agency in
the public, commercial or not-for-profit sectors.
\section*{Acknowledgements and use of AI-assisted tools}
During manuscript preparation, the authors used OpenAI Codex for language editing, theorem-structure review and LaTeX reformatting. The authors are responsible for verifying all mathematical statements, empirical results, citations and final wording.

\appendix
\renewcommand{\thesection}{\Alph{section}}
\renewcommand{\theHsection}{appendix.\Alph{section}}
\section{Proof and measurability details}\label{sec:appendix}

\subsection{Proof of Theorem~\ref{thm:factorization}}

\begin{proof}
	Fix an arbitrary Borel set \(B\subseteq\calT\), and write
	\[
	E:=h^{-1}(B).
	\]
	Since \(h\) is Borel, both \(E\) and \(\calY\setminus E\) are Borel subsets of
	\(\calY\). We first verify explicitly that \(E\) is saturated with respect to
	the fibres of \(M\), that is,
	\[
	M^{-1}\{M(E)\}=E.
	\]
	The inclusion \(E\subseteq M^{-1}\{M(E)\}\) is immediate. Conversely, suppose
	that \(y\in M^{-1}\{M(E)\}\). Then there exists \(z\in E\) such that
	\(M(y)=M(z)\). By maximality of \(M\), there is some \(g\in G\) for which
	\(y=g\cdot z\). Invariance of \(h\) therefore gives
	\[
	h(y)=h(g\cdot z)=h(z)\in B,
	\]
	and hence \(y\in E\). This proves the asserted saturation. Notice that this
	statement concerns the particular set \(E=h^{-1}(B)\): it says that every
	fibre of \(M\) is either entirely contained in \(E\) or entirely contained in
	its complement; it does not assert that the fibres of \(M\) are singletons.
	
	We next claim that
	\[
	M(E)\cap M(\calY\setminus E)=\varnothing.
	\]
	Indeed, if \(s\) belonged to this intersection, there would exist
	\(y\in E\) and \(z\in\calY\setminus E\) such that
	\[
	M(y)=s=M(z).
	\]
	Maximality of \(M\) would then place \(y\) and \(z\) in the same \(G\)-orbit,
	so invariance of \(h\) would imply \(h(y)=h(z)\). This is impossible because
	\(h(y)\in B\), whereas \(h(z)\notin B\).
	
	Because \(\calY\) and \(\calS\) are standard Borel spaces, and \(M\) is
	Borel, the images \(M(E)\) and \(M(\calY\setminus E)\) are analytic subsets
	of \(\calS\). They need not themselves be Borel, which is why a separation
	argument is required. By Lusin's separation theorem, there exists a Borel
	set \(D_B\subseteq\calS\) satisfying
	\[
	M(E)\subseteq D_B
	\qquad\text{and}\qquad
	M(\calY\setminus E)\subseteq\calS\setminus D_B.
	\]
	We now show that
	\[
	E=M^{-1}(D_B).
	\]
	If \(y\in E\), then \(M(y)\in M(E)\subseteq D_B\), and hence
	\(y\in M^{-1}(D_B)\). If \(y\notin E\), then
	\(M(y)\in M(\calY\setminus E)\subseteq\calS\setminus D_B\), and hence
	\(y\notin M^{-1}(D_B)\). The two inclusions therefore follow, and
	\[
	h^{-1}(B)=E=M^{-1}(D_B)\in\sigma(M).
	\]
	
	Since \(B\subseteq\calT\) was an arbitrary Borel set, \(h\) is
	\(\sigma(M)\)-measurable. The measurable factorization lemma, applicable
	because \(\calT\) is a standard Borel space, therefore yields a Borel map
	\[
	\bar h:\calS\longrightarrow\calT
	\]
	such that
	\[
	h=\bar h\circ M.
	\]
	
	It remains to establish the stated uniqueness. Suppose that
	\(\bar h_1,\bar h_2:\calS\to\calT\) are two Borel maps satisfying
	\[
	h=\bar h_1\circ M=\bar h_2\circ M.
	\]
	For any \(s\in M(\calY)\), choose \(y\in\calY\) with \(M(y)=s\). Then
	\[
	\bar h_1(s)
	=\bar h_1\{M(y)\}
	=h(y)
	=\bar h_2\{M(y)\}
	=\bar h_2(s).
	\]
	Thus \(\bar h_1=\bar h_2\) on \(M(\calY)\). No uniqueness is claimed outside
	\(M(\calY)\), since values of a factor map there do not affect its
	composition with \(M\).
\end{proof}

\subsection{Existence constructions for maximal invariants}

If a compact metrizable group acts continuously on a Polish space, the orbit-valued map $y\mapsto G\cdot y$ is a Borel maximal invariant taking values in the hyperspace of nonempty compact subsets with its Vietoris Borel structure. For a finite group $H=\{h_1,\ldots,h_m\}$ acting by Borel maps, any Borel injection $\nu:\calY\to[0,1]$ gives the canonical representative
\[
c_H(y)=h_{j^\star}\cdot y,\qquad j^\star=\min\argmin_j\nu(h_j\cdot y),
\]
which is Borel and maximal. The lattice construction combines an explicit cross-section for noncompact translations with this finite residual minimization.

\subsection{Measurability of the Haar orbit kernel}

For compact second-countable $G$, a jointly Borel action and measurable section make $m\mapsto K_m(B)$ in~\eqref{eq:haar-kernel} measurable for every Borel $B$ by integration of the jointly measurable indicator $(g,m)\mapsto\1\{g\cdot s(m)\in B\}$. Stabilizers may change the map from $G$ onto an orbit but do not change the pushforward orbit measure.

\section*{Figure caption list}
\noindent Figure~\ref{fig:simulation}. Simulation rejection fractions across intrinsic effect strengths; causal validity under informative acquisition pertains to invariant endpoints.

\section*{Table caption list}
\noindent Table~\ref{tab:simulation}. Simulation rejection fractions.\\
Table~\ref{tab:primary}. Primary RxRx1 contrast.\\
Table~\ref{tab:secondary}. Prespecified secondary contrasts.\\
Table~\ref{tab:sharpnull}. Acquisition-only sharp-null stress test.\\
Table~\ref{tab:approx}. Approximate-action stress test.


\begin{thebibliography}{99}

\bibitem{Bhattacharjee2025}
Bhattacharjee, S., Li, B., Wu, X. and Xue, L. (2025) Doubly robust estimation of causal effects for random object outcomes with continuous treatments. \emph{arXiv:2506.22754}.

\bibitem{Blackwell1953}
Blackwell, D. (1953) Equivalent comparisons of experiments. \emph{Ann. Math. Statist.}, 24, 265--272.

\bibitem{ClopperPearson1934}
Clopper, C. J. and Pearson, E. S. (1934) The use of confidence or fiducial limits illustrated in the case of the binomial. \emph{Biometrika}, 26, 404--413.

\bibitem{Curry2022}
Curry, J., Mukherjee, S. and Turner, K. (2022) How many directions determine a shape and other sufficiency results for two topological transforms. \emph{Trans. Amer. Math. Soc. Ser. B}, 9, 1006--1043.

\bibitem{Eaton1989}
Eaton, M. L. (1989) \emph{Group Invariance Applications in Statistics}. Institute of Mathematical Statistics, Hayward, CA.

\bibitem{Fisher1935}
Fisher, R. A. (1935) \emph{The Design of Experiments}. Oliver and Boyd, Edinburgh.

\bibitem{Gretton2012}
Gretton, A., Borgwardt, K. M., Rasch, M. J., Sch\"olkopf, B. and Smola, A. (2012) A kernel two-sample test. \emph{J. Mach. Learn. Res.}, 13, 723--773.

\bibitem{Holm1979}
Holm, S. (1979) A simple sequentially rejective multiple test procedure. \emph{Scand. J. Stat.}, 6, 65--70.

\bibitem{Kallenberg2021}
Kallenberg, O. (2021) \emph{Foundations of Modern Probability}, 3rd edn. Springer, Cham.

\bibitem{Kechris1995}
Kechris, A. S. (1995) \emph{Classical Descriptive Set Theory}. Springer, New York.

\bibitem{Kurisu2024}
Kurisu, D., Zhou, Y., Otsu, T. and M\"uller, H.-G. (2024) Geodesic causal inference. \emph{arXiv:2406.19604}.

\bibitem{LeCam1986}
Le Cam, L. (1986) \emph{Asymptotic Methods in Statistical Decision Theory}. Springer, New York.

\bibitem{LehmannRomano2005}
Lehmann, E. L. and Romano, J. P. (2005) \emph{Testing Statistical Hypotheses}, 3rd edn. Springer, New York.

\bibitem{Otsu1979}
Otsu, N. (1979) A threshold selection method from gray-level histograms. \emph{IEEE Trans. Syst. Man Cybern.}, 9, 62--66.

\bibitem{Raykov2025}
Raykov, Y. P., Luo, H., Strait, J. D. and KhudaBukhsh, W. R. (2025) Kernel-based estimators for functional causal effects. \emph{arXiv:2503.05024}.

\bibitem{Recursion2023}
[dataset] Recursion (2023) RxRx1: an image set for cellular morphological variation across many biological perturbations. \url{https://www.rxrx.ai/rxrx1}.

\bibitem{Robins1986}
Robins, J. M. (1986) A new approach to causal inference in mortality studies with a sustained exposure period. \emph{Math. Modelling}, 7, 1393--1512.

\bibitem{RosenbaumRubin1983}
Rosenbaum, P. R. and Rubin, D. B. (1983) The central role of the propensity score in observational studies for causal effects. \emph{Biometrika}, 70, 41--55.

\bibitem{Rubin1974}
Rubin, D. B. (1974) Estimating causal effects of treatments in randomized and nonrandomized studies. \emph{J. Educ. Psychol.}, 66, 688--701.

\bibitem{Shin2024}
Shin, H.-Y., Kim, K., Lee, K. and Oh, H.-S. (2024) Absolute average and median treatment effects as causal estimands on metric spaces. \emph{arXiv:2407.03726}.

\bibitem{Sriperumbudur2010}
Sriperumbudur, B. K., Gretton, A., Fukumizu, K., Sch\"olkopf, B. and Lanckriet, G. R. G. (2010) Hilbert space embeddings and metrics on probability measures. \emph{J. Mach. Learn. Res.}, 11, 1517--1561.

\bibitem{Sypetkowski2023}
Sypetkowski, M. et al. (2023) RxRx1: a dataset for evaluating experimental batch correction methods. \emph{arXiv:2301.05768}.

\bibitem{Torgersen1991}
Torgersen, E. (1991) \emph{Comparison of Statistical Experiments}. Cambridge University Press, Cambridge.

\bibitem{Turner2014}
Turner, K., Mukherjee, S. and Boyer, D. M. (2014) Persistent homology transform for modeling shapes and surfaces. \emph{Information and Inference}, 3, 310--344.

\bibitem{Wijsman1990}
Wijsman, R. A. (1990) \emph{Invariant Measures on Groups and Their Use in Statistics}. Institute of Mathematical Statistics, Hayward, CA.

\end{thebibliography}
\end{document}